\documentclass{elsarticle}
\usepackage{amsmath,amsthm,amsfonts,thm-restate,hyperref,cleveref}

\newtheorem{definition}{Definition}
\newtheorem{proposition}[definition]{Proposition}
\newtheorem{theorem}[definition]{Theorem}
\newtheorem{lemma}[definition]{Lemma}

\crefname{equation}{}{}
\crefname{lemma}{Lemma}{Lemmas}
\crefname{proposition}{Proposition}{Propositions}
\crefname{section}{Section}{Sections}
\crefname{table}{Table}{Tables}
\crefname{theorem}{Theorem}{Theorems}

\newcommand{\A}{\mathcal{A}}

\newcommand{\N}{\mathbb{N}}

\journal{Information Processing Letters}

\begin{document}


\title{On Complementing Unambiguous Automata and Graphs With Many Cliques and Cocliques}
\author{Emil Indzhev}
\author{Stefan Kiefer\texorpdfstring{\fnref{mylabel}}{}}
\fntext[mylabel]{Stefan Kiefer is supported by a Royal Society University Research Fellowship.}
\address{University of Oxford, UK}


\begin{abstract}
We show that for any unambiguous finite automaton with $n$ states there exists an unambiguous finite automaton with $\sqrt{n+1} \cdot 2^{n/2}$ states that recognizes the complement language.
This builds and improves upon a similar result by Jir{\'{a}}sek et al.~[Int.\ J.\ Found.\ Comput.\ Sci.\ 29~(5) (2018)].
Our improvement is based on a reduction to and an analysis of a problem from extremal graph theory:
we show that for any graph with $n$ vertices, the product of the number of its cliques with the number of its cocliques (independent sets) is bounded by $(n+1) 2^n$.
\end{abstract}

\maketitle


\section{Introduction} \label{sec:intro}

Given two finite automata $\A_1, \A_2$, recognizing languages $L_1, L_2 \subseteq \Sigma^*$, respectively, the \emph{state complexity} of union (or intersection, or complement, etc.) is how many states (in terms of the number of states in $\A_1, \A_2$) may be needed in the worst case for an automaton that recognizes $L_1 \cup L_2$ (or $L_1 \cap L_2$, or $\Sigma^* \setminus L_1$, etc.).
The state complexity depends on the type of automaton considered, such as nondeterministic finite automata (NFAs), deterministic finite automata (DFAs), or unambiguous finite automata (UFAs).
UFAs are NFAs that, for each word $w \in \Sigma^*$, have either one or zero accepting runs.

The state complexity has been well studied for various types of automata and language operations, see, e.g., \cite{JirasekJS18} and the references therein for some known results.
For example, complementing an NFA with $n$~states may require $2^n$ states~\cite{Birget93}, even for automata with binary alphabet~\cite{Jiraskova05}.
However, the state complexity of UFA complementation is not yet fully understood.
It was shown only in~2018 by Raskin~\cite{Raskin18} that it is not polynomial.
\begin{proposition}[\cite{Raskin18}] \label{prop:Raskin18}
For any $n \in \N$ there exists a unary (i.e., $|\Sigma|=1$) UFA $\A_n$ with $n$~states such that any NFA that recognizes $\Sigma^* \setminus L(\A_n)$ has at least $n^{(\log \log \log n)^{\Theta(1)}}$ states.
\end{proposition}
This super-polynomial blowup (even for unary alphabet and even if the output automaton is allowed to be ambiguous) refuted a conjecture that it may be possible to complement UFAs with a polynomial blowup~\cite{Colcombet15}.
A non-trivial upper bound (for general alphabets and outputting a UFA) 
 was shown by Jir{\'{a}}sek et al.~\cite{JirasekJS18}.
\begin{proposition}[\cite{JirasekJS18}] \label{prop:Jirasek}
Let $\A$ be a UFA with $n \ge 7$ states that recognizes a language $L \subseteq \Sigma^*$.
Then there exists a UFA with at most $n \cdot 2^{0.786 n}$ states that recognizes the language $\Sigma^* \setminus L$.
\end{proposition}
In this note we build and improve on this result to obtain the following theorem.
\begin{restatable}{theorem}{thmmain} \label{thm:main}
Let $\A$ be a UFA with $n \ge 0$ states that recognizes a language $L \subseteq \Sigma^*$.
Then there exists a UFA with at most $\sqrt{n+1} \cdot 2^{n/2}$ states that recognizes the language $\Sigma^* \setminus L$.
\end{restatable}
\Cref{thm:main} is based on a tight analysis of the complementation construction due to Jir{\'{a}}sek et al.~\cite{JirasekJS18}.
Their construction performs two UFA complementation procedures and picks the smaller UFA among the two.
The first procedure is the standard subset construction for determinizing NFAs, followed by swapping accepting and non-accepting states.
The second procedure is a symmetrical ``backward'' variant.
It was shown in~\cite{JirasekJS18} that the smaller of the two resulting UFAs has at most $n \cdot 2^{0.786 n}$ states.
We improve this upper bound to $\sqrt{n+1} \cdot 2^{n/2}$ states.
We also show that this bound (on the complexity of this particular construction) is tight up to a constant factor~$2$; i.e., for every $n$ there is a UFA with $n$~states where the complementation construction from~\cite{JirasekJS18} produces a UFA with at least $\frac12 \sqrt{n+1} \cdot 2^{n/2}$ states.

Our improvement is based on two technical contributions which may be of independent interest:
\begin{enumerate}
\item
We reduce the analysis of the complementation construction described above to a problem from extremal graph theory: loosely speaking, we show in \cref{sec:reduction} that applying the construction to a UFA with $n$~states can yield a large UFA if and only if there exists a graph with $n$~vertices that has both many cliques and many cocliques (independent sets).
\item
We solve this graph problem: we show in \cref{sec:graph} that any graph with $n$~vertices has at most $\sqrt{n+1} \cdot 2^{n/2}$ cliques or at most $\sqrt{n+1} \cdot 2^{n/2}$ cocliques.
This bound is tight up to a factor of~$2$.
\end{enumerate}
These results enable us to prove \cref{thm:main} in \cref{sec:main}.
There we also show that our analysis of the complementation construction from~\cite{JirasekJS18} is almost tight.

\section{Preliminaries} \label{sec:prelims}

\subsection*{Finite Automata}
A \emph{nondeterministic finite automaton (NFA)} is a quintuple $\A = (Q,\Sigma,\delta,I,F)$, where $Q$ is the finite set of states, $\Sigma$ is the finite alphabet, $\delta \subseteq Q \times \Sigma \times Q$ is the transition relation, $I \subseteq Q$ is the set of initial states, and $F \subseteq Q$ is the set of accepting states.
We write $q \xrightarrow{a} r$ to denote that $(q,a,r) \in \delta$.
A finite sequence $q_0 \xrightarrow{a_1} q_1 \xrightarrow{a_2} \cdots \xrightarrow{a_n} q_n$ is called a \emph{run} and can also be summarized as $q_0 \xrightarrow{a_1 a_2 \cdots a_n} q_n$.
The NFA~$\A$ \emph{recognizes} the language $L(\A) := \{w \in \Sigma^* \mid \exists\, q_0 \in I \,.\, \exists\, f \in F \,.\, q_0 \xrightarrow{w} f\}$.
The NFA~$\A$ is a \emph{deterministic finite automaton (DFA)} if $|I| = 1$ and for every $q \in Q$ and $a \in \Sigma$ there is exactly one $q'$ with $q \xrightarrow{a} q'$.
The NFA~$\A$ is an \emph{unambiguous finite automaton (UFA)} if for every word $w = a_1 a_2 \cdots a_n \in \Sigma^*$ there is at most one \emph{accepting} run for~$w$, i.e., a run $q_0 \xrightarrow{a_1}  q_1 \xrightarrow{a_2} \cdots \xrightarrow{a_n} q_n$ with $q_0 \in I$ and $q_n \in F$.
Clearly, every DFA is a UFA.

Let $\A = (Q,\Sigma,\delta,I,F)$ be an NFA.
For $S \subseteq Q$ and $w \in \Sigma^*$ we write
\begin{align*}
\delta(S,w)\ &:=\ \{r \in Q \mid \exists\,q \in S \,.\, q \xrightarrow{w} r\}\,, \quad \text{ and}\\
\delta^{-1}(w,S)\  &:=\ \{r \in Q \mid \exists\,q \in S \,.\, r \xrightarrow{w} q\}\,.
\end{align*}
The \emph{forward determinization} of~$\A$ is the DFA obtained by the standard subset construction, i.e., the DFA $\A_f := (Q_f,\Sigma,\delta_f, \{I\}, F_f)$ with $Q_f := \{ \delta(I,w) \mid w \in \Sigma^* \}$ and $\delta_f := \{(S,a,\delta(S,a)) \mid S \in Q_f,\ a \in \Sigma\}$ and $F_f := \{S \in Q_f \mid S \cap F \ne \emptyset\}$.
Analogously, the \emph{backward determinization} of~$\A$ is the NFA  $\A_b := (Q_b,\Sigma,\delta_b, I_b, \{F\})$ where $Q_b := \{ \delta^{-1}(w,F) \mid w \in \Sigma^* \}$ and $\delta_b := \{(\delta^{-1}(a,S),a,S) \mid S \in Q_b,\ a \in \Sigma\}$ and $I_b := \{S \in Q_b \mid S \cap I \ne \emptyset\}$.
Note that $L(\A) = L(\A_f) = L(\A_b)$.
The NFA $\A_b$ is ``backward-deterministic'', i.e., it has a single accepting state and for every $S \in Q_b$ and $a \in \Sigma$ there is exactly one $S'$ with $S' \xrightarrow{a} S$.
It follows that $\A_b$ is a UFA.
UFAs recognizing $\Sigma^* \setminus L(\A)$ can be obtained from $\A_f$ (resp., $\A_b$) by swapping accepting and non-accepting (resp., initial and non-initial) states.
This implies the following lemma.%
\begin{lemma} \label{lem:NFA-compl}
Let $\A$ be an NFA that recognizes a language $L \subseteq \Sigma^*$.
Suppose that its forward determinization has $k$ states and its backward determinization has $\ell$ states.
Then there is a UFA with at most $\min\{k, \ell\}$ states that recognizes the language $\Sigma^* \setminus L$.
\end{lemma}
Like \cref{prop:Jirasek} from~\cite{JirasekJS18}, \cref{thm:main} is based on \cref{lem:NFA-compl}.

\subsection*{Graphs}
An (undirected, simple, finite) \emph{graph} is a pair $G = (V,E)$, where $V$ is the finite set of vertices and $E \subseteq \{\{v,v'\} \subseteq V \mid v \ne v'\}$ is the set of edges.
A \emph{clique} in~$G$ is a set $X \subseteq V$ of vertices such that whenever $v, v' \in X$ and $v \ne v'$ then $\{v,v'\} \in E$.
A \emph{coclique} (or independent set) in~$G$ is a set $Y \subseteq V$ of vertices such that whenever $v, v' \in Y$ and $v \ne v'$ then $\{v,v'\} \not\in E$.
Note that any subset of a clique, including the empty set, is also a clique, and similarly for cocliques.

\section{Reduction to a Graph Problem} \label{sec:reduction}

The proof of our main result, \cref{thm:main}, rests on two key auxiliary results.
In this section we prove the first key lemma, \cref{lem:key-reduction}.
Loosely speaking, it says that a UFA with a large forward determinization and a large backward determinization defines a graph with many cliques and cocliques.

\begin{lemma} \label{lem:key-reduction}
Let $\A = (Q,\Sigma,\delta,I,F)$ be a UFA.
Suppose that its forward determinization has $k$ states, and its backward determinization has $\ell$ states.
Then there exists a graph $G = (Q, E)$ that has at least $k$~cliques and at least $\ell$~cocliques.
\end{lemma}
\begin{proof}
Define
\[
 E \ := \ \{\{q, q'\} \subseteq Q \mid q \ne q',\ \exists\, q_0, q_0' \in I \,.\, \exists\, w \in \Sigma^* \,.\, q_0 \xrightarrow{w} q \land q_0' \xrightarrow{w} q'\}\;.
\]
That is, $\{q, q'\}$ is an edge in~$G$ if the states $q, q'$ in~$\A$ are reachable from initial states via the same word.
It follows from the definition of the forward determinization that every state in the forward determinization of~$\A$ is a clique in~$G$.
So $G$ has at least $k$ cliques.

Let $Y \subseteq Q$ be a state in the backward determinization of~$\A$.
It suffices to show that $Y$ is a coclique in~$G$.
If $Y = \emptyset$, then $Y$ is a coclique.
Let $q, q' \in Y$.
It suffices to show that $\{q,q'\}\not\in E$.
If there is no word $w_1 \in \Sigma^*$ with $q,q' \in \delta(I,w_1)$, then $\{q,q'\} \not\in E$.
So we can assume that there are states $q_0, q_0' \in I$ and a word $w_1 \in \Sigma^*$ such that $q_0 \xrightarrow{w_1} q$ and $q_0' \xrightarrow{w_1} q'$.
Since $Y \supseteq \{q,q'\}$ is a state in the backward determinization, there are a word $w_2 \in \Sigma^*$ and states $f, f' \in F$ such that $q \xrightarrow{w_2} f$ and $q' \xrightarrow{w_2} f'$ in~$\A$.
Thus, $\A$ accepts the word $w_1 w_2$ via runs
\[
 q_0 \xrightarrow{w_1} q \xrightarrow{w_2} f \quad \text{ and } \quad q_0' \xrightarrow{w_1} q' \xrightarrow{w_2} f'\,.
\]
Since $\A$ is unambiguous, these runs are equal.
Thus, $q = q'$.
We conclude that $\{q,q'\} \not\in E$.
\end{proof}

\Cref{lem:key-reduction} has the following converse.
\begin{lemma} \label{lem:converse}
Let $(V,E)$ be a graph with $k$~cliques and $\ell$~cocliques.
Then there is a UFA with $|V|$~states whose forward determinization has at least $k$~states and whose backward determinization has at least $\ell$~states.
\end{lemma}
\begin{proof}
Let $(V,E)$ be a graph, and let $\mathcal{X}, \mathcal{Y} \subseteq 2^V$ be the sets of its cliques and cocliques, respectively.
If $V = \emptyset$, then $\mathcal{X} = \mathcal{Y} = \{\emptyset\}$, and $\emptyset$ is a state (the only one) of the forward and the backward determinization of a UFA with no states.
So we can assume that there is $v_0 \in V$.
Define the NFA $\A := (V,\Sigma,\delta,\{v_0\},\{v_0\})$ where $\Sigma := \Sigma_1 \cup \Sigma_2$ and $\Sigma_1 := \{1\} \times \mathcal{X}$ and $\Sigma_2 := \{2\} \times \mathcal{Y}$ and $\delta := \delta_1 \cup \delta_2$ and $\delta_1 := \{(v_0, (1,X), v) \mid v \in X \in \mathcal{X}\}$ and $\delta_2 := \{(v, (2,Y), v_0) \mid v \in Y \in \mathcal{Y}\}$.
For each $X \in \mathcal{X}$ we have $\delta(\{v_0\}, (1,X)) = X$, and for each $Y \in \mathcal{Y}$ we have $\delta^{-1}((2,Y),\{v_0\}) = Y$.
Hence, every clique is a state of the forward determinization of~$\A$, and every coclique is a state of the backward determinization of~$\A$.

It remains to show that $\A$ is a UFA.
Let $w_1, w_2 \in \Sigma^*$ and $v, v' \in V$ with $v_0 \xrightarrow{w_1} v \xrightarrow{w_2} v_0$ and $v_0 \xrightarrow{w_1} v' \xrightarrow{w_2} v_0$.
It suffices to show that $v = v'$.
Note that all transitions labelled with~$(1,X)$ have $v_0$ as source, and all transitions labelled with~$(2,Y)$ have $v_0$ as target.
If $w_1$ or~$w_2$ is the empty word or $w_1$ ends with a letter of the form $(2,Y)$ or $w_2$ begins with a letter of the form $(1,X)$, then $v = v_0 = v'$.
So we can assume that $w_1$ ends with, say, $(1,X)$, and $w_2$ begins with, say, $(2,Y)$.
We have
\begin{align*}
v,v' \ &\in \ \delta(\{v_0\}, w_1) \ \subseteq \ \delta(\{v_0\},(1,X)) \ = \ X \qquad\quad \text{ and} \\
v,v' \ &\in \ \delta^{-1}(w_2,\{v_0\}) \ \subseteq \ \delta^{-1}((2,Y),\{v_0\}) \ = \ Y\,.
\end{align*}
Since $X$ is a clique, we have $v = v'$ or $\{v,v'\} \in E$.
Since $Y$ is a coclique, we have $v = v'$ or $\{v,v'\} \not\in E$.
We conclude that $v = v'$.
\end{proof}

\section{Graphs With Many Cliques and Many Cocliques} \label{sec:graph}

This section is about a problem from extremal graph theory.
We show that any graph with $n$~vertices has at most $\sqrt{n+1} \cdot 2^{n/2}$ cliques or at most $\sqrt{n+1} \cdot 2^{n/2}$ cocliques,
and we present a graph that has at least $\frac12 \sqrt{n+1} \cdot 2^{n/2}$ cliques and at least $\frac12 \sqrt{n+1} \cdot 2^{n/2}$ cocliques; see \cref{thm:key-graph} below.

Our proof strategy is to consider pairs of a clique and a coclique.
The following lemma counts the possibilities of partitioning the vertices of a subgraph into a clique and a coclique.
\begin{lemma} \label{lem:boundPS}
Let $(V,E)$ be a graph.
For any $S \subseteq V$, let
 \[
  P_S \ := \ \{X \subseteq S \mid X \text{ is a clique and } S \setminus X \text{ is a coclique}\}\,.
 \]
Then $|P_S| \le |S| + 1$.
\end{lemma}
\begin{proof}
We proceed by induction on~$|S|$.
For $S=\emptyset$ we have $P_\emptyset = \{\emptyset\}$.
Suppose that $|P_S| \le |S| + 1$ holds for some $S \subset V$, and let $v \in V \setminus S$.
For the inductive step, it suffices to show that $|P'| \le |P_S| + 1$ where $P' := P_{S \cup \{v\}}$.
Every $X' \in P'$ can be expressed as $X$ or as $X \cup \{v\}$ for some $X \in P_S$.

Suppose that for some $X \subseteq S$ we have both $X \in P'$ and $X \cup \{v\} \in P'$.
We claim that $X$ must be the set of neighbours of~$v$ in~$S$.
Indeed, since $X \cup \{v\} \in P'$ is a clique, all vertices in~$X$ are neighbours of~$v$ in~$S$.
Conversely, since $X \in P'$, the set $(S \cup \{v\}) \setminus X$ is a coclique; thus, every neighbour of~$v$ in~$S$ is in~$X$.

Hence, there is at most one $X \in P_S$ with both $X \in P'$ and $X \cup \{v\} \in P'$.
It follows that $|P'| \le |P_S| + 1$.
\end{proof}

The following lemma is similar, but considers pairs of a clique and a coclique that may overlap.
\begin{lemma} \label{lem:boundRS}
Let $(V,E)$ be a graph.
For any $S \subseteq V$, let
 \[
  R_S \ := \ \{(X,Y) \in 2^S \times 2^S \mid X \cup Y = S \text{ and } X \text{ is a clique and } Y \text{ is a coclique}\}\,.
 \]
Then $|R_S| \le 2 |S| + 1$.
\end{lemma}
\begin{proof}
Let $S \subseteq V$.
Every $(X,Y) \in R_S$ with $X \cap Y = \emptyset$ can be uniquely written as $(X, S \setminus X)$ with $X \in P_S$, where $P_S$ is defined as in \cref{lem:boundPS}.
By \cref{lem:boundPS}, there are at most $|S| + 1$ pairs $(X,Y) \in R_S$ with $X \cap Y = \emptyset$.

Let $v \in S$.
Suppose there is $(X_v,Y_v) \in R_S$ with $v \in X_v \cap Y_v$.
Since $X_v$ is a clique, every vertex in $X_v \setminus \{v\}$ is a neighbour of~$v$ in~$S$.
Since $Y_v$ is a coclique, the neighbours of~$v$ in~$S$ are not in~$Y_v$, but since $X_v \cup Y_v = S$ they must be in~$X_v$.
Thus, $X_v \setminus \{v\}$ is the set of neighbours of~$v$ in~$S$.
By a symmetric argument, $Y_v \setminus \{v\}$ is the set of non-neighbours of~$v$ in~$S$.
It follows that $(X_v,Y_v)$ is the only $(X,Y) \in R_S$ with $v \in X \cap Y$.

We conclude that there are at most $|S|$ (one for each $v \in S$) pairs $(X,Y) \in R_S$ with $X \cap Y \ne \emptyset$.
Hence, $|R_S| \le (|S| + 1) + |S| = 2 |S| + 1$.
\end{proof}
\Cref{lem:boundRS} implies a bound on the product of the number of cliques with the number of cocliques.
\begin{lemma} \label{lem:boundT}
Let $(V,E)$ be a graph with $|V| = n$.
Then
\[
 |\{X \subseteq V \mid X \text{ is a clique}\}| \cdot |\{Y \subseteq V \mid Y \text{ is a coclique}\}| \ \le \ (n+1) 2^n\,.
\]
\end{lemma}
\begin{proof}
The product on the left-hand side is equal to~$|R|$, where
\[
 R \ := \ \{(X,Y) \in 2^V \times 2^V \mid  X \text{ is a clique and } Y \text{ is a coclique}\}\,.
\]
We have $R = \bigcup_{S \subseteq V} R_S$, so by \cref{lem:boundRS} we have
\begin{align*}
 |R| \ &\le\ \sum_{i=0}^{n} \binom{n}{i} (2 i + 1) \\
       &=\ \sum_{i=0}^n \left( \binom{n}{i} i + \binom{n}{i} i + \binom{n}{i} \right) \\
       &=\ \sum_{i=0}^n \left( \binom{n}{i} i + \binom{n}{i} (n-i) + \binom{n}{i} \right) \\
       &=\ \sum_{i=0}^n \left( \binom{n}{i} n + \binom{n}{i} \right) \ = \ (n+1) 2^n \,. \tag*{\qedhere}
\end{align*}
\end{proof}
The bound in \cref{lem:boundT} is tight, as a \emph{complete} graph with $n$ vertices (which has all $\binom{n}{2}$ possible edges) has $2^n$ cliques and $n+1$ cocliques.

Now we can prove the following theorem, which serves as the second of our two key auxiliary results, but may be of independent interest.
\begin{theorem} \label{thm:key-graph}
Any graph with $n$ vertices has at most $\sqrt{n+1} \cdot 2^{n/2}$ cliques or at most $\sqrt{n+1} \cdot 2^{n/2}$ cocliques.
Moreover, for any $n \ge 0$ there is a graph with $n$~vertices that has at least $\frac12 \sqrt{n+1} \cdot 2^{n/2}$ cliques and at least $\frac12 \sqrt{n+1} \cdot 2^{n/2}$ cocliques.
\end{theorem}

\begin{proof}
Since $\min\{k,\ell\} \le \sqrt{k\cdot \ell}$ holds for all $k,\ell$, the upper bound follows from \cref{lem:boundT}.
Towards the lower bound, if $n = 0$, the lower bound holds, as the empty set is a clique and a coclique.
Let $n \ge 1$.
Define a graph $G = (V_1 \cup V_2, E)$ with $|V_1| = k$ and $|V_2| = n-k$ for some $k \in \{0, 1, \ldots, n\}$ and $E = \{\{v,v'\} \subseteq V_1 \mid v \ne v'\}$.
Graph~$G$ has at least $2^k$ cliques and at least (in fact, exactly) $(k+1) 2^{n-k}$ cocliques.
Choose $k$ as an integer nearest to $\frac{n}{2} + \frac12 \log_2 \frac{n+1}{2}$.
Then we have
\begin{align*}
 2^k \ &\ge\ 2^{\frac{n}{2} - \frac12 + \frac12 \log_2 \frac{n+1}{2}} \ = \ \frac12 \sqrt{n+1} \cdot 2^{\frac{n}{2}} \qquad\qquad\quad \text{ and} \\
 (k+1) 2^{n-k} \ &\ge\ \left( \frac{n}{2} + \frac12 + \frac12 \log_2 \frac{n+1}{2} \right) 2^{\frac{n}{2} - \frac12 - \frac12 \log_2 \frac{n+1}{2}} \\
 &\ge\ \frac{n+1}{2} \cdot 2^{\frac{n}{2} - \frac12 - \frac12 \log_2 \frac{n+1}{2}} \\
 &=\ \frac12 \sqrt{n+1} \cdot 2^{\frac{n}{2}}\,. \tag*{\qedhere}
\end{align*}
\end{proof}

\section{Proof of the Main Result and Conclusions} \label{sec:main}

In the previous two sections we proved our key auxiliary results, \cref{lem:key-reduction,thm:key-graph}, respectively.
We use them to show the main theorem, \cref{thm:main}, which is restated here for convenience.

\thmmain*

\begin{proof}
Suppose that the forward determinization of~$\A$ has $k$ states, and the backward determinization of~$\A$ has $\ell$ states.
By \cref{lem:key-reduction}, there is a graph with $n$ vertices, at least $k$ cliques and at least $\ell$ cocliques.
By \cref{thm:key-graph}, it follows that $k \le \sqrt{n+1} \cdot 2^{n/2}$ or $\ell \le \sqrt{n+1} \cdot 2^{n/2}$.
Hence, by \cref{lem:NFA-compl}, there is a UFA with at most $\sqrt{n+1} \cdot 2^{n/2}$ states that recognizes the language $\Sigma^* \setminus L$.
\end{proof}

By the argument justifying \cref{lem:NFA-compl}, the UFA for $\Sigma^* \setminus L$ in \cref{thm:main} is obtained either by swapping accepting and non-accepting states in the forward determinization or by swapping initial and non-initial states in the backward determinization.
So we have actually proved that for any UFA with $n$~states, either its forward determinization or its backward determinization has at most $\sqrt{n+1} \cdot 2^{n/2}$ states.
The next proposition shows that this bound is optimal up to a factor of~$2$.
\begin{proposition} \label{prop:tight-Jirasek}
For every $n \ge 0$ there is a UFA with $n$~states such that both its forward and its backward determinization have at least $\frac12 \sqrt{n+1} \cdot 2^{n/2}$ states.
\end{proposition}
\begin{proof}
Let $n \ge 0$.
By \cref{thm:key-graph}, there is a graph with $n$~vertices that has at least $\frac12 \sqrt{n+1} \cdot 2^{n/2}$ cliques and at least $\frac12 \sqrt{n+1} \cdot 2^{n/2}$ cocliques.
Now the statement follows from \cref{lem:converse}.
\end{proof}
We remark that the UFA from \cref{prop:tight-Jirasek}, based on the construction from the proof of \cref{lem:converse}, has an alphabet whose size is the number of cliques plus the number of cocliques, i.e., at least $\sqrt{n+1} \cdot 2^{n/2}$.

Overall, we conclude that the UFA complementation construction due to Jir{\'{a}}sek et al.~\cite{JirasekJS18} has a worst-case state complexity of $\Theta(\sqrt{n+1} \cdot 2^{n/2})$.
This is currently the best upper bound on the state complexity of complementing UFAs.
Obtaining it has been the main contribution of this article.
It remains possible that other constructions lead to smaller, even sub-exponential, UFAs for the complement language.

%
%


\bibliography{references}

\end{document}